\documentclass[a4paper,11pt]{article}
\usepackage{geometry}
\usepackage{microtype}
\usepackage{authblk}
\usepackage{amssymb}
\usepackage{amsthm}
\usepackage{amsmath}
\usepackage{hyperref}
\newtheorem{theorem}{Theorem}[section]
\newtheorem{lemma}[theorem]{Lemma}
\newtheorem{corollary}[theorem]{Corollary}

\theoremstyle{definition}
\newtheorem{example}[theorem]{Example}

\newcommand{\Z}{\mathbb Z}
\newcommand{\eps}{\varepsilon}
\DeclareMathOperator{\pex}{Exp}
\newcommand{\gex}{\overline \pex}
\newcommand{\mor}{\mathcal H}
\newcommand{\nonper}{\mathcal N}
\newcommand{\inj}{\mathcal I}
\newcommand{\sat}{\mathrm{EqSat}}
\newcommand{\satcf}{\mathrm{EqSatCF}}
\newcommand{\pow}{\mathrm{Pow}}
\newcommand{\nonprim}{\mathrm{NonPrim}}
\newcommand{\res}[1]{#1_\diamond}
\newcommand{\ext}[1]{#1^\diamond}
\DeclareMathOperator{\alphabet}{alph}

\title{Mapping words to powers by morphisms}
\author{Aleksi Saarela\thanks{Supported by the Research Council of Finland under grant 339311}}
\affil{Department of Mathematics and Statistics, University of Turku, Finland\\
\texttt{amsaar@utu.fi}}

\begin{document}

\maketitle

\begin{abstract}
We characterize the words that can be mapped to arbitrarily high powers by injective morphisms. For all other words, we prove a linear upper bound for the highest power that they can be mapped to, and this bound is optimal up to a constant factor if there is no restriction on the size of the alphabet. We also prove that, for any integer $n \geq 2$, deciding whether a given word can be mapped to an $n$th power by a nonperiodic morphism is NP-hard and in PSPACE, and so is deciding whether a given word can be mapped to a nonprimitive word by a nonperiodic morphism.
\end{abstract}

\section{Introduction}

We study the following question:
Given a word and a family of morphism,
to which powers can the word be mapped by the morphisms in the family?
We are especially interested in the highest power a word can be mapped into,
and in the algorithmic question of whether a word can be mapped to an $n$th power,
where $n \geq 2$ is either fixed or can be freely chosen by the algorithm.
The most natural families of morphisms to consider
are nonperiodic morphisms and injective morphisms.

In the article~\cite{mi97},
the morphisms that map all primitive words to primitive words are studied,
and there are connections to square-free morphism
and therefore to the classical topic of repetition-freeness of infinite words.
Some of the questions we study in this article might seem a bit similar,
but instead of looking at images of many words under a given morphism,
we look at images of a given word under many morphisms.

One source of motivation for this article is that
both powers and morphisms are fundamental concepts,
and the questions we study have very simple formulations,
which suggests that the questions are inherently interesting.
Moreover, powers and morphisms appear very frequently in connection to many other topics,
so the results in this article might find a use
when studying other questions in the area of combinatorics on words.
As an example,
knowing that certain words cannot be mapped to nonprimitive words by nonperiodic morphisms
has previously been useful in some constructions, such as in~\cite[Lemma 15]{sa20icalp}.

More specific motivation comes from the theory of word equations.
For example, a word $w$ can be mapped to a square by a nonperiodic morphism
if and only if
there exists a nonperiodic solution for the constant-free word equation $(w, x^2)$,
where we treat the letters of $w$ as variables
and where $x$ is a new variable that does not appear in $w$.
The same idea of course works for higher powers.
This immediately leads to a connection between our topic
and word equations where one side of the equation is a power of a variable.
These kinds of equations have been studied in many articles,
such as~\cite{lysc62,le65,apdj68,hako97,ho01,hano05,sa19jcta}.
There is also a less obvious connection to the satisfiability problem of word equations.
Some recent results on the satisfiability problem and its different variants
can be found in~\cite{je22,dagahemano18,sa20icalp}, for example.
These connections allow us to turn some known results on word equations
into results on mapping words to powers.
They also open up the possibility of studying some variants of the satisfiability problem
through arguably simpler decision problems
such as the problem of checking whether a word can be mapped to a square by a nonperiodic morphism.

Finally, there is a more algebraic source of motivation.
As an introductory example from classic algebra,
most integers are not squares of other integers.
However, given such an integer $n$,
we can always generate the ring $\Z[\sqrt{n}]$,
and in this ring, $n$ is a square.
However, the ring $\Z[\sqrt{n}]$ does not necessarily have all the nice properties
that the ring of integers has.
For example, it is not necessarily a unique factorization domain (UFD).
Studying the question of when $\Z[\sqrt{n}]$ is a UFD leads to very deep results in algebra.
Analogously, given a word $w$ that is not a square,
we can find an extension monoid of words
such that $w$ is a square in this extension monoid.
The concept of a free monoid can be seen as analogous to the concept of a UFD,
and we can ask the question of when can the extension monoid be chosen so that it is free.
This question is the same as asking whether
$w$ can be mapped to a square by an injective morphism.

The results in this article are divided into three sections.
In Section~\ref{sec:basic},
we prove some very basic results about mapping words to powers.
For example, we study whether the powers to which a word can be mapped
depend on the choice of alphabets or not.

In Section~\ref{sec:high},
we study the highest powers to which words can be mapped by injective morphisms.
We characterize the words that can be mapped to arbitrarily high powers.
For all other words,
we prove an upper bound that is linear with respect to the length of the word.
Moreover, we prove that this upper bound is optimal up to a constant factor
if all words over arbitrarily large alphabets are considered.
For binary words, there is a constant upper bound.
The behavior over ternary and larger alphabets remains an open question.

In Section~\ref{sec:algo},
we prove that the question of whether a given word can be mapped by some nonperiodic morphism
to a square, or to any other fixed power,
is polynomially equivalent to the nonperiodic satisfiability problem of word equations.
Consequently, this problem is NP-hard and in PSPACE.
The same is true for the problem of deciding whether a given word
can be mapped to a nonprimitive word.

\section{Preliminaries}

We start by recalling some standard notation, definitions, and results
related to combinatorics on words.
For more, see~\cite{chka97,lo02}.

We use the symbols $\Sigma, \Gamma$ to denote alphabets.
Alphabets are assumed to be nonunary and finite unless otherwise mentioned.
The empty word is denoted by $\eps$.
The length of a word $w$ is denoted by $|w|$
and the number of occurrences of a letter $a$ in $w$ is denoted by $|w|_a$.
The set of all letters that occur in $w$ is denoted by $\alphabet(w)$.

A word $u$ is a \emph{factor} of a word $v$
if there exist words $x, y$ such that $v = xuy$.
If we can choose $x = \eps$, then $u$ is a \emph{prefix},
if we can choose $y = \eps$, then $u$ is a \emph{suffix},
and if we can choose $x \ne \eps \ne y$, then $u$ is an \emph{internal factor} of $v$.
Words $u, v$ are \emph{conjugates}
if there exist words $x, y$ such that $u = xy$ and $v = yx$.

A nonempty word is \emph{primitive} if it is not a power of a shorter word,
and it is \emph{nonprimitive} otherwise.
If $u = v^n$ and $v$ is primitive, then $v$ is a \emph{primitive root} of $u$.
The following facts are well-known:
\begin{itemize}
    \item
    Every nonempty word has a unique primitive root.
    \item
    Words $u, v$ commute, that is, $uv = vu$,
    if and only if they are powers of a common word.
    \item
    Nonempty words commute
    if and only if they have the same primitive root.
    \item
    A conjugate of a primitive word is primitive.
    \item
    If $u$ is a primitive word, then $u$ is not an internal factor of $u^2$.
\end{itemize}

A mapping $h: \Sigma^* \to \Gamma^*$ is a \emph{morphism}
if $h(xy) = h(x)h(y)$ for all $x, y \in \Sigma^*$.
A morphism $h$ is \emph{periodic}
if there exists a word $r$ such that $h(x) \in r^*$ for all $x \in \Sigma$,
or equivalently if $h(xy) = h(yx)$ for all $x, y \in \Sigma$.
It is well-known that if $\Sigma$ is a binary alphabet,
then every nonperiodic morphism $\Sigma^* \to \Gamma^*$ is injective.

Let $\Sigma$ be an alphabet of variables and $\Gamma$ an alphabet of constants.
A \emph{constant-free word equation}
is a pair of words $(u, v) \in \Sigma^* \times \Sigma^*$.
A \emph{solution} of this equation is a morphism
\begin{math}
    h: \Sigma^* \to \Gamma^*
\end{math}
such that $h(u) = h(v)$.
The equation $(u, v)$
can be called an \emph{equation over $\Sigma$} or \emph{equation over $(\Sigma, \Gamma)$}.

The above definition can be extended to allow the equation to have constants:
A \emph{word equation with constants}
is a pair of words $(u, v) \in (\Sigma \cup \Gamma)^* \times (\Sigma \cup \Gamma)^*$.
A \emph{solution} of this equation is a morphism
\begin{math}
    h: \Sigma^* \to \Gamma^*
\end{math}
such that $\ext{h}(u) = \ext{h}(v)$,
where $\ext{h}: (\Sigma \cup \Gamma)^* \to \Gamma^*$ is the extension of $h$
defined by $\ext{h}(a) = a$ for all $a \in \Gamma$.
Obviously, we need to assume here that $\Sigma$ and $\Gamma$ are disjoint.

A \emph{system of equations} is a set of equations.
A \emph{solution} of a system is a morphism that is a solution of every equation in the system.
The \emph{length} of an equation $(u, v)$ is $|uv|$,
and the length of a system of equations is the sum of the lengths of the equations in the system.

Checking whether a given system of equations with constants has a solution
can be done in nondeterministic linear space~\cite{je22}.
On the other hand, it is NP-hard.
Checking whether a given system of constant-free equations has a nonperiodic solution
is also NP-hard~\cite{sa20icalp}.

\begin{example}
    The constant-free word equation $(x^2, yzy)$ over $(\{x, y, z\}, \Gamma)$
    has nonperiodic solutions $h$ defined by
    \begin{equation*}
        h(x) = (pq)^{i + 1} p,\qquad
        h(y) = (pq)^i p,\qquad
        h(z) = qppq
    \end{equation*}
    for all $p, q \in \Gamma^*$, $pq \ne qp$, $i \in \Z_{\geq 0}$,
    and periodic solutions $h$ defined by
    \begin{equation*}
        h(x) = p^{i + j},\qquad
        h(y) = p^i,\qquad
        h(z) = p^{2j}
    \end{equation*}
    for all $p \in \Gamma^*$, $i, j \in \Z_{\geq 0}$.
\end{example}

The following result is known as the theorem of Lyndon and Sch\"utzenberger~\cite{lysc62}.

\begin{theorem}
    Let $\Sigma = \{x, y, z\}$ and $k, m, n \geq 2$.
    The constant-free word equation $(x^k, y^m z^n)$ has only periodic solutions.
\end{theorem}

Another way to formulate this theorem is that if $h$ is a nonperiodic morphism,
then $h(y^m z^n)$ is primitive.
The next theorem gives a generalization of this result.
It was proved by Spehner~\cite{sp76}
and independently by Barbin-Le Rest and Le Rest \cite{bale85}.
Different proofs can be found in~\cite{sa18fi} and~\cite{horast}.

\begin{theorem} \label{thm:morph-prim}
    Let $h: \{y, z\}^* \to \Gamma^*$ be a nonperiodic morphism and let
    \begin{equation*}
        S = \{w \in \{y, z\}^+ \smallsetminus \{y, z\}
            \mid \text{$w$ is primitive and $h(w)$ is nonprimitive}\}.
    \end{equation*}
    Either $S = \varnothing$
    or $S$ is the conjugacy class of some primitive word $w$ such that $|w|_y = 1$ or $|w|_z = 1$.
\end{theorem}

To conclude this section, we recall the periodicity theorem of Fine and Wilf~\cite{fiwi65}.

\begin{theorem}
    Let $u, v \in \Sigma^+$ and $m, n \geq 1$.
    If $u^m$ and $v^n$ have a common prefix of length
    \begin{math}
        |uv| - \gcd(|u|, |v|),
    \end{math}
    then $u$ and $v$ are powers of a common word of length $\gcd(|u|, |v|)$.
\end{theorem}

\section{Basic results} \label{sec:basic}

Let $\Sigma$ and $\Gamma$ be alphabets and $H$ a family of morphisms $\Sigma^* \to \Gamma^*$.
We are interested in the relation $h(w) = r^n$,
where $h \in H$, $w \in \Sigma^+$, $r \in \Gamma^+$, $n \in \Z_+$,
and sometimes we might want to add a requirement that $r$ is primitive.
More specifically,
we are interested in the possible values of the exponent $n$ for a fixed word $w$.
We define the sets of such exponents
\begin{align*}
    \pex_H(w) &=
    \{n \in \Z_+ \mid \exists h \in H, r \in \Sigma^+: h(w) = r^n \text{ and $r$ is primitive}\},\\
    \gex_H(w) &=
    \{n \in \Z_+ \mid \exists h \in H, r \in \Sigma^+: h(w) = r^n\}.
\end{align*}
The following families of morphisms are very natural:
\begin{itemize}
    \item
    The family of all morphisms $\Sigma^* \to \Gamma^*$,
    denoted by $\mor(\Sigma, \Gamma)$.
    \item
    The family of nonperiodic morphisms $\Sigma^* \to \Gamma^*$,
    denoted by $\nonper(\Sigma, \Gamma)$.
    \item
    The family of injective morphisms $\Sigma^* \to \Gamma^*$,
    denoted by $\inj(\Sigma, \Gamma)$.
\end{itemize}
If $\Sigma$ and $\Gamma$ are clear from context, or if the choice of alphabets does not matter,
we can denote $\mor(\Sigma, \Gamma)$, $\nonper(\Sigma, \Gamma)$, $\inj(\Sigma, \Gamma)$
by just $\mor$, $\nonper$, $\inj$, respectively.

\begin{example} \label{exa:aabb}
    Let $\Sigma = \Gamma = \{a, b\}$.
    In this example,
    we determine $\pex_\nonper(w)$ and $\gex_\nonper(w)$ for $w \in \{aa, aab, aabb\}$.
    For all $n \in \Z_+$,
    let $h_n$ be the morphism defined by $h_n(a) = a^n$, $h_n(b) = b$,
    and let $g_n$ be the morphism defined by $g_n(a) = a$, $g_n(b) = b(aab)^{n - 1}$.
    We have $h_n, g_n \in \nonper$.
    Because $h_n(aa) = (aa)^n$, the word $aa$ can be mapped to an $n$th power for all $n$.
    However, $h(aa)$ is a square for all morphisms $h$,
    so $aa$ can only be mapped to even powers of primitive words.
    Therefore
    \begin{equation*}
        \pex_\nonper(aa) = 2 \Z_+, \qquad \gex_\nonper(aa) = \Z_+.
    \end{equation*}
    The word $aab$ can be mapped to any power of a primitive word: $g_n(aab) = (aab)^n$.
    Thus
    \begin{equation*}
        \pex_\nonper(aab) = \gex_\nonper(aab) = \Z_+.
    \end{equation*}
    The word $aabb$, on the other hand,
    cannot be mapped by a nonperiodic morphism to an $n$th power for any $n \geq 2$, that is,
    \begin{equation*}
        \pex_\nonper(aabb) = \gex_\nonper(aabb) = \{1\}.
    \end{equation*}
    This follows from the theorem of Lyndon and Sch\"utzenberger.
\end{example}

The set $\gex_H(w)$ is completely determined by the set $\pex_H(w)$:
\begin{equation} \label{eq:gexpex}
    \gex_H(w) = \{d \in \Z_+ \mid \exists n \in \pex_H(w): d|n\}.
\end{equation}
In this section we prove several results of the type that
\begin{equation*}
    \pex_H(w) = \pex_{H'}(w')
    \qquad \text{and} \qquad
    \gex_H(w) = \gex_{H'}(w')
\end{equation*}
for some families of morphisms $H, H'$ and words $w, w'$.
By \eqref{eq:gexpex}, the second of these two equalities follows from the first,
so it is sufficient to prove the first one.
The converse is not true.
For instance, we saw in Example~\ref{exa:aabb} that
$\gex_\nonper(aa) = \gex_\nonper(aab)$ but $\pex_\nonper(aa) \ne \pex_\nonper(aab)$.

By the next theorem,
we can usually concentrate on the case where $w$ is primitive
when studying the sets $\pex_H(w)$ and $\gex_H(w)$,

\begin{theorem} \label{thm:prim-power}
    Let $H$ be a family of morphisms
    and let $w = r^k$ where $r$ is a primitive word.
    Then
    \begin{align*}
        \pex_H(w) &= \{kn \in \Z_+ \mid n \in \pex_H(r)\},\\
        \gex_H(w) &= \{k'n' \in \Z_+ \mid n' \in \gex_H(r) \text{ and } k'|k\}.
    \end{align*}
\end{theorem}

\begin{proof}
    First, we prove the claim about $\pex_H(w)$.
    If $n \in \pex_H(r)$, then $h(r) = s^n$ for some $h \in H$ and primitive $s$,
    and then
    \begin{equation*}
        h(w) = h(r)^k = s^{kn},
    \end{equation*}
    so $kn \in \pex_H(w)$.
    On the other hand,
    if $m \in \pex_H(w)$, then $h(w) = t^m$ for some $h \in H$ and primitive $t$,
    and we can write $h(r) = s^n$ for some $n \in \pex_H(r)$ and primitive $s$.
    Then
    \begin{equation*}
        t^m = h(w) = h(r)^k = s^{kn},
    \end{equation*}
    so it must be $t = s$ and $m = kn$.

    Next, we prove the claim about $\gex_H(w)$.
    If $n' \in \gex_H(r)$ and $k = jk'$ for some integers $j$ and $k'$,
    then $h(r) = u^{n'}$ for some $h \in H$ and word $u$,
    and then
    \begin{equation*}
        h(w) = h(r)^k = u^{kn'} = (u^j)^{k'n'},
    \end{equation*}
    so $k'n' \in \gex_H(w)$.
    On the other hand,
    if $m \in \gex_H(w)$, then $h(w) = u^m$ for some $h \in H$ and word $u$,
    and we can write $h(r) = s^n$ for some $n \in \pex_H(r)$ and primitive $s$.
    Then
    \begin{equation*}
        u^m = h(w) = h(r)^k = s^{kn},
    \end{equation*}
    so it must be $m|kn$.
    It follows that we can write $m = k'n'$, where $k'|k$ and $n'|n$,
    and $n' \in \gex_H(r)$ by \eqref{eq:gexpex}.
\end{proof}

%
%

The next theorem shows that the family of all morphisms is not very interesting:
The set $\gex_{\mor}(w)$ is completely trivial,
and determining $\pex_{\mor}(w)$
can be reduced to determining $\pex_{\nonper}(w)$.

\begin{theorem} \label{thm:mor}
    Let $\Sigma = \{a_1, \dots, a_m\}$ and $\Gamma$ be alphabets.
    For every word $w \in \Sigma^+$, we have
    \begin{align*}
        \gex_{\mor}(w) &= \Z_+,\\
        \pex_{\mor}(w) &= \pex_{\nonper}(w) \cup
            \Big\{\sum_{i = 1}^m k_i |w|_{a_i} \mid k_1, \dots, k_m \in \Z_{\geq 0}\Big\}
            \smallsetminus \{0\}.
    \end{align*}
\end{theorem}

\begin{proof}
    For any $n \in \Z_+$, we can define a morphism $h$ by $h(a_i)= a^n$ for all $i$,
    and then $h(w) = (a^{|w|})^n$, so $n \in \gex_{\mor}(w)$.

    We have $h \in \mor \smallsetminus \nonper$
    if and only if
    there exists a primitive word $r$ and nonnegative integers $k_i$
    such that $h(a_i)= r^{k_i}$ for all $i$.
    Then $h(w) = r^n$, where $n = \sum_{i = 1}^m k_i |w|_{a_i}$,
    so the set $\pex_{\mor \smallsetminus \nonper}(w)$
    consists of exactly these numbers $n$, except 0.
\end{proof}

The next theorem shows that, in the case of the families $\mor, \nonper, \inj$,
the choice of the target alphabet $\Gamma$ is not significant.

\begin{theorem} \label{thm:gamma}
    Let $H \in \{\mor, \nonper, \inj\}$
    and let $\Sigma, \Gamma_1, \Gamma_2$ be alphabets.
    Let $w \in \Sigma^+$.
    Then
    \begin{equation*}
        \pex_{H(\Sigma, \Gamma_1)}(w) = \pex_{H(\Sigma, \Gamma_2)}(w)
        \qquad \text{and} \qquad
        \gex_{H(\Sigma, \Gamma_1)}(w) = \gex_{H(\Sigma, \Gamma_2)}(w).
    \end{equation*}
\end{theorem}

\begin{proof}
    Let $\Gamma_1 = \{a_1, \dots, a_m\}$ and let $a, b \in \Gamma_2$ be distinct.
    Let $g: \Gamma_1 \to \Gamma_2$ be the morphism defined by $g(a_i) = a^i b$ for all $i$.
    Then $g$ is injective.
    Moreover, we can see that $g$ maps primitive words to primitive words:
    If $g(x) = y^k$ for some words $x, y$ and positive integer $k$,
    then $y$ ends with $b$, so $y = g(z)$ for some word $z$,
    and then from $g(x) = g(z^k)$ and the injectivity of $g$ it follows that $x = z^k$.
    This means that the primitivity of $x$ implies the primitivity of $g(x)$.

    Let $n \in \pex_{H(\Sigma, \Gamma_1)}(w)$
    and let $h(w) = r^n$ for some $h \in H(\Sigma, \Gamma_1)$ and primitive word $r$.
    We have $g(h(w)) = g(r)^n$ and $g(r)$ is primitive.
    Moreover, if $h$ is nonperiodic, then so is $g \circ h$,
    and if $h$ is injective, then so is $g \circ h$,
    so in all cases, $g \circ h \in H(\Sigma, \Gamma_2)$.
    Therefore    $n \in \pex_{H(\Sigma, \Gamma_2)}(w)$.
    This shows that
    \begin{math}
        \pex_{H(\Sigma, \Gamma_1)}(w) \subseteq \pex_{H(\Sigma, \Gamma_2)}(w),
    \end{math}
    and the other direction is of course symmetric.
\end{proof}

In the case of $\nonper(\Sigma, \Gamma)$,
the choice of the alphabet $\Sigma$ can make a difference,
but for a rather trivial reason:
If $\Sigma$ is larger than $\alphabet(w)$,
then any periodic morphism can be turned into a nonperiodic morphism
without changing the image of $w$.
This leads to the following result.

\begin{theorem} \label{thm:nonper-sigma}
    Let $w \in \Sigma^+$.
    If $\alphabet(w) \ne \Sigma = \{a_1, \dots, a_m\}$, then
    \begin{align*}
        \gex_{\nonper(\Sigma, \Gamma)}(w) &= \Z_+,\\
        \pex_{\nonper(\Sigma, \Gamma)}(w)
        &= \pex_{\nonper(\alphabet(w), \Gamma)}(w)
            \cup \Big\{\sum_{i = 1}^m k_i |w|_{a_i} \mid k_1, \dots, k_m \in \Z_{\geq 0}\Big\}
            \smallsetminus \{0\}.
    \end{align*}
\end{theorem}

\begin{proof}
    The set of images $h(w)$ for $h \in \nonper(\Sigma, \Gamma)$
    is exactly the set of images $\res{h}(w)$
    where $\res{h} \in \mor(\alphabet(w), \Gamma)$
    is a restriction of $h \in \nonper(\Sigma, \Gamma)$.
    The set of these restrictions is actually the whole set $\mor(\alphabet(w), \Gamma)$,
    with the possible exception of the morphism that maps $w$ to $\eps$.
    This is because any $g \in \mor(\alphabet(w), \Gamma)$ such that $g(w) \ne \eps$
    can be extended to a morphism $\ext{g} \in \nonper(\Sigma, \Gamma)$
    by defining $\ext{g}(a)$ for $a \in \Sigma \smallsetminus \alphabet(w)$
    so that it does not commute with $g(b)$ for some $b \in \alphabet(w)$.
    Therefore
    \begin{equation*}
        \pex_{\nonper(\Sigma, \Gamma)}(w) = \pex_{\mor(\alphabet(w), \Gamma)}(w)
        \qquad \text{and} \qquad
        \gex_{\nonper(\Sigma, \Gamma)}(w) = \gex_{\mor(\alphabet(w), \Gamma)}(w).
    \end{equation*}
    The claim now follow from Theorem~\ref{thm:mor}.
\end{proof}

If $w$ is fixed, then the choice of $\Sigma$ in $\inj(\Sigma, \Gamma)$ does not matter,
as long as $w \in \Sigma^+$ of course.
This is proved in the next theorem.
It should be noted, however, that if $w$ is not fixed,
and instead we look at the sets $\pex_{\inj(\Sigma, \Gamma)}(w)$ as $w$ runs over $\Sigma^+$,
then the size of $\Sigma$ can make a difference.

\begin{theorem} \label{thm:inj-sigma}
    Let $w \in \Sigma^+$.
    Then
    \begin{equation*}
        \pex_{\inj(\Sigma, \Gamma)}(w) = \pex_{\inj(\alphabet(w), \Gamma)}(w)
        \qquad \text{and} \qquad
        \gex_{\inj(\Sigma, \Gamma)}(w) = \gex_{\inj(\alphabet(w), \Gamma)}(w).
    \end{equation*}
\end{theorem}

\begin{proof}
    The set of images $h(w)$ for $h \in \inj(\Sigma, \Gamma)$
    is exactly the set of images $\res{h}(w)$
    where $\res{h} \in \mor(\alphabet(w), \Gamma)$ is a restriction of $h$.
    The set of these restrictions is a subset of $\inj(\alphabet(w), \Gamma)$,
    because a restriction of an injective morphism is still injective.
    Therefore
    \begin{equation*}
        \pex_{\inj(\Sigma, \Gamma)}(w) \subseteq \pex_{\inj(\alphabet(w), \Gamma)}(w).
    \end{equation*}

    On the other hand, a morphism in $\inj(\alphabet(w), \Gamma)$
    is not necessarily a restriction of any morphism in $\inj(\Sigma, \Gamma)$,
    but it is a restriction of some morphism in $\inj(\Sigma, \Gamma')$ for some $\Gamma'$:
    For every $a \in \Gamma$, we can define a new letter $a'$,
    and let $\Gamma' = \Gamma \cup \{a' \mid a \in \Gamma\}$,
    and then we can extend every $h \in \inj(\alphabet(w), \Gamma)$
    to a morphism $\ext{h} \in \inj(\Sigma, \Gamma')$
    by defining $\ext{h}(a) = a'$ for all $a \in \Sigma \smallsetminus \alphabet(w)$.
    Thus
    \begin{equation*}
        \pex_{\inj(\alphabet(w), \Gamma)}(w)
        \subseteq \pex_{\inj(\Sigma, \Gamma')}(w)
        = \pex_{\inj(\Sigma, \Gamma)}(w),
    \end{equation*}
    where the last equality follows from Theorem~\ref{thm:gamma}.
    This completes the proof.
\end{proof}

The next two theorems give connections between word equations and mapping words to powers.

\begin{theorem} \label{thm:gex-eq}
    Let $w \in \Sigma^+$, $n \in \Z_+$, and let $x \notin \Sigma$ be a letter.
    Then $n \in \gex_{\nonper(\Sigma, \Gamma)}(w)$ if and only if
    the constant-free word equation $(w, x^n)$ over $(\Sigma \cup \{x\}, \Gamma)$
    has a nonperiodic solution.
\end{theorem}

\begin{proof}
    If $n \in \gex_{\nonper(\Sigma, \Gamma)}(w)$,
    then there exists $h \in \nonper(\Sigma, \Gamma)$
    such that $h(w) = u^n$ for some $u \in \Gamma^+$.
    The morphism $h$ can be extended to $\ext{h} \in \nonper(\Sigma \cup \{x\}, \Gamma)$
    by defining $\ext{h}(x) = u$,
    and then $\ext{h}$ is a nonperiodic solution of the equation $(w, x^n)$.

    On the other hand,
    if $h \in \nonper(\Sigma \cup \{x\}, \Gamma)$
    is a nonperiodic solution of the equation $(w, x^n)$,
    then we can let $\res{h} \in \mor(\Sigma, \Gamma)$ be a restriction of $h$.
    If $\res{h}$ is periodic,
    then there exists a primitive word $r$ such that $\res{h}(a) \in r^*$ for all $a \in \Sigma$,
    and then $h(x)^n = h(w) = \res{h}(w) \in r^*$
    and thus $h(x) \in r^*$,
    meaning that also $h$ is periodic, which is a contradiction.
    Thus it must be $\res{h} \in \nonper(\Sigma, \Gamma)$,
    and from $\res{h}(w) = h(x)^n$
    it follows that $n \in \gex_{\nonper(\Sigma, \Gamma)}(w)$,
    except in the trivial case $h(x) = \eps$.
    But if $h(x) = \eps$, then it must be $h(a) = \eps$ for all $a \in \alphabet(w)$,
    which means that $h$ being nonperiodic
    implies that $\Sigma \smallsetminus \alphabet(w) \ne \varnothing$,
    and then $n \in \gex_\nonper(w)$ by Theorem~\ref{thm:nonper-sigma}.
\end{proof}

\begin{theorem} \label{thm:nonprim-eq}
    Let $w \in \Sigma^+$ and let $x, y \notin \Sigma$ be letters.
    There exists $n \in \gex_{\nonper(\Sigma, \Gamma)}(w)$, $n \geq 2$,
    if and only if
    the system of constant-free word equations $\{(w, x^2 y^3), (xy, yx)\}$
    over $(\Sigma \cup \{x, y\}, \Gamma)$
    has a nonperiodic solution.
\end{theorem}

\begin{proof}
    If there exists $n \in \gex_{\nonper(\Sigma, \Gamma)}(w)$, $n \geq 2$,
    then $h(w) = u^n$ for some $u \in \Gamma^+$ and $h \in \nonper(\Sigma, \Gamma)$.
    We can write $n = 2i + 3j$ for some nonnegative integers $i, j$.
    The morphism $h$ can be extended to $\ext{h} \in \nonper(\Sigma \cup \{x, y\}, \Gamma)$
    by defining $\ext{h}(x) = u^i$ and $\ext{h}(y) = u^j$,
    and then $\ext{h}(x^2 y^3) = u^n = \ext{h}(w)$ and $\ext{h}(xy) = \ext{h}(yx)$,
    so $\ext{h}$ is a nonperiodic solution of the system.

    On the other hand,
    if $h \in \nonper(\Sigma \cup \{x, y\}, \Gamma)$
    is a nonperiodic solution of the system,
    then $h(x)$ and $h(y)$ commute and are thus powers of some common word $u \in \Gamma^+$,
    say $h(x) = u^i$ and $h(y) = u^j$,
    and then $h(w) = u^{2i + 3j}$.
    Let $\res{h} \in \mor(\Sigma, \Gamma)$ be a restriction of $h$.
    If $\res{h}$ is periodic,
    then there exists a primitive word $r$ such that $\res{h}(a) \in r^*$ for all $a \in \Sigma$,
    and then $u^{2i + 3j} = h(w) = \res{h}(w) \in r^*$
    and thus $u \in r^*$ and $h(x), h(y) \in r^*$,
    meaning that also $h$ is periodic, which is a contradiction.
    Thus it must be $\res{h} \in \nonper(\Sigma, \Gamma)$,
    and from $\res{h}(w) = u^{2i + 3j}$
    it follows that $2i + 3j \in \gex_\nonper(w)$, $2i + 3j \geq 2$,
    except in the trivial case $i = j = 0$.
    But if $i = j = 0$, then it must be $h(a) = \eps$ for all $a \in \alphabet(w)$,
    which means that $h$ being nonperiodic
    implies that $\Sigma \smallsetminus \alphabet(w) \ne \varnothing$,
    and then $\gex_\nonper(w) = \Z_+$ by Theorem~\ref{thm:nonper-sigma}.
\end{proof}

The next theorem proves a claim that was made in the introduction.

\begin{theorem} \label{thm:embed}
    Let $w \in \Sigma^+$.
    Then $\Sigma^*$ can be embedded into a free monoid
    so that $w$ becomes an $n$th power in that monoid
    if and only if
    $n \in \gex_\inj(w)$.
\end{theorem}

\begin{proof}
    Let $M$ be a free monoid and $g: \Sigma^* \to M$ an embedding
    such that $g(w) = u^n$ for some $u \in M$.
    We can assume that $M = \Delta^*$ for some possibly infinite alphabet $\Delta$.
    There exists a finite set $\Gamma \subseteq \Delta$ such that $u \in \Gamma^+$.
    Let $h \in \inj(\alphabet(w), \Gamma)$ be a restriction of $g$.
    Then $h(w) = u^n$, so $n \in \gex_\inj(w)$.

    Conversely, if $n \in \gex_\inj(w)$,
    then $h(w) = v^n$ for some $\Gamma$, $h \in \inj(\Sigma, \Gamma)$, $v \in \Gamma^+$.
    Here $h$ is an embedding of $\Sigma^*$ into the free monoid $\Gamma^*$
    and $h(w)$ is an $n$th power.
\end{proof}

\section{High powers} \label{sec:high}

In this section, we study the following question: When is $\pex_\inj(w)$ infinite,
and if it is finite, then how large can the largest element be?
By \eqref{eq:gexpex},
$\pex_\inj(w)$ and $\gex_\inj(w)$ are either both infinite
or have the same largest element,
so it does not matter whether we use $\pex_\inj(w)$ or $\gex_\inj(w)$ here.
We concentrate on the case where $w$ is primitive.
Nonprimitive words can then be handled with Theorem~\ref{thm:prim-power}.

In the future, the same questions could be studied also for $\nonper$ in place of $\inj$,
but here we concentrate on injective morphisms
because they have some very nice properties,
such as Theorems~\ref{thm:inj-sigma} and \ref{thm:embed},
and because we can prove strong results about $\pex_\inj(w)$.

First, we prove that primitive words that contain some letter exactly once
can be mapped to arbitrarily high powers by injective morphisms.

\begin{theorem} \label{thm:inj-inf}
    Let $w \in \Sigma^+$ be a primitive word.
    If there exists a letter $a \in \Sigma$ such that $|w|_a = 1$, then
    \begin{equation*}
        \pex_\inj(w) = \Z_+.
    \end{equation*}
\end{theorem}

\begin{proof}
    We can write $w = uav$, where $u, v \in (\Sigma \smallsetminus \{a\})^*$.
    For all $n \in \Z_+$,
    we can define a morphism $h: \Sigma^* \to \Sigma^*$ by $h(a) = a(vua)^{n - 1}$
    and $h(b) = b$ for all letters $b \in \Sigma \smallsetminus \{a\}$.
    Then $h$ is injective and
    \begin{equation*}
        h(w) = u a(vua)^{n - 1} v = (uav)^n,
    \end{equation*}
    so $n \in \pex_\inj(w)$.
\end{proof}

For all other primitive words,
we get an upper bound that is linear with respect to the length of the word.

\begin{theorem} \label{thm:inj-ub}
    Let $w \in \Sigma^+$ be a primitive word.
    If $|w|_a \ne 1$ for all $a \in \Sigma$, then
    \begin{equation*}
        \max(\pex_\inj(w)) < |w|.
    \end{equation*}
\end{theorem}

\begin{proof}
    We assume that for some morphism $h$,
    $h(w) = u^n$, where $u$ is primitive and $n \geq |w|$,
    and show that $h$ is not injective.
    If $|h(a)| < |u|$ for all $a \in \alphabet(w)$,
    then
    \begin{equation*}
        |h(w)| < |w||u| \leq n|u| = |u^n|,
    \end{equation*}
    which is a contradiction.
    Thus there necessarily exists $a \in \alphabet(w)$ such that $|h(a)| \geq |u|$.
    It must be $|w|_a \geq 2$,
    so some conjugate of $w$ can be written as $asat$, where $s, t \in \Sigma^*$.
    Then $h(asat) = v^n$ for some conjugate $v$ of $u$.
    We can write $h(as) = v^k  p$, $h(at) = q v^{n - k - 1}$, where $v = pq$.
    Because $|h(a)| \geq |u| = |v|$, $h(a)$ begins with $pq$ and also with $qp$,
    so $pq = qp$ and thus $p$ and $q$ are powers of a common word.
    By the primitivity of $v$, either $p = v$, $q = \eps$ or $p = \eps$, $q = v$.
    Thus $h(as), h(at) \in v^*$.
    It follows that $h(asat) = h(atas)$.
    Because $asat$ is primitive, it cannot be an internal factor of $(asat)^2$,
    but $atas$ is an internal factor of $(asat)^2$, so it must be $asat \ne atas$.
    Therefore $h$ is not injective.
\end{proof}

To summarize, Theorems~\ref{thm:inj-inf} and \ref{thm:inj-ub} show that
$\pex_\inj(w)$ is infinite
if and only if
some letter occurs in the primitive root of $w$ exactly once,
and for all other words $w$, $\max(\pex_\inj(w)) < |w|$.
The next theorem shows that if the size of the alphabet $\Sigma$ is not bounded,
then the upper bound in Theorem~\ref{thm:inj-ub} is optimal up to a constant factor.

\begin{theorem} \label{thm:inj-lb}
    Let $n$ be even and let $\Sigma$ be an $n$-ary alphabet.
    There exists a primitive word $w \in \Sigma^{2n}$
    such that $|w|_c = 2$ for all $c \in \Sigma$ and
    \begin{equation*}
        \max(\pex_\inj(w)) \geq \frac{|w|}{2} - 1.
    \end{equation*}
\end{theorem}

\begin{proof}
    Let $\Sigma = \{a_1, \dots, a_n\}$ and
    \begin{equation*}
        w = \Big(\prod_{i = 1}^{n - 2} a_i^2\Big) \cdot a_{n - 1} a_n^2 a_{n - 1}.
    \end{equation*}
    We define a morphism $h: \Sigma^* \to \{a, b\}^*$ by
    \begin{equation*}
        h(a_i) = a^{2n - 2i - 2} b a^{2i} \text{ for $i \in \{1, \dots, n - 2\}$},\qquad
        h(a_{n - 1}) = b,\qquad
        h(a_n) = a^{n - 1}.
    \end{equation*}
    Then
    \begin{align*}
        h(w)
        &= \Big(\prod_{i = 1}^{n - 2} a^{2n - 2i - 2} b a^{2n - 2} b a^{2i}\Big)
            \cdot b a^{2n - 2} b\\
        &= \Big(\prod_{i = 1}^{n - 2} a^{2i - 2} a^{2n - 2i - 2} b a^{2n - 2} b\Big)
            \cdot a^{2n - 4} b a^{2n - 2} b
        = (a^{2n - 4} b a^{2n - 2} b)^{n - 1}.
    \end{align*}
    To show that $h$ is injective,
    we assume that $h(a_i u) = h(a_j v)$ for some indices $i < j$ and words $u, v$
    and derive a contradiction.
    It must be $1 \leq i \leq n - 2$ and $j = n$,
    because otherwise the first $b$ in $h(a_i u)$ and the first $b$ in $h(a_j v)$
    would be in a different position.
    Then we get
    \begin{equation} \label{eq:inj-lb1}
        a^{n - 2i - 1} b a^{2i} h(u) = h(v)
    \end{equation}
    by cancelling $a^{n - 1}$ from the beginning of both sides of $h(a_i u) = h(a_j v)$.
    Let $a_k$ be the first letter of $v$.
    If $k = n$,
    then the left-hand side of \eqref{eq:inj-lb1} begins with $a^{n - 2i - 1} b$
    and the right-hand side with $a^{n - 1}$,
    which is a contradiction because $n - 2i - 1 < n - 1$.
    But if $k \ne n$,
    then the left-hand side of \eqref{eq:inj-lb1} begins with $a^{n - 2i - 1} b$
    and the right-hand side with $a^m b$ for some even $m$,
    which is a contradiction because $n - 2i - 1$ is odd.
    This shows that $h$ is injective.
    Thus $n - 1 \in \pex_\inj(w)$.
\end{proof}

\begin{example}
    We illustrate Theorem~\ref{thm:inj-lb} and its proof in the case $n = 4$.
    For convenience, let us denote the letters $a_1, a_2, a_3, a_4$ by $a, b, c, d$, respectively.
    The word $w$ in the proof is $w = a^2 b^2 c d^2 c$
    and the morphism $h$ is defined by
    \begin{equation*}
        h(a) = a^4 b a^2,\qquad
        h(b) = a^2 b a^4,\qquad
        h(c) = b,\qquad
        h(d) = a^3.
    \end{equation*}
    We have
    \begin{equation*}
        h(w)
        = a^4 b a^2 \cdot a^4 b a^2
            \cdot a^2 b a^4 \cdot a^2 b a^4
            \cdot b \cdot a^3 \cdot a^3 \cdot b
        = a^4 b a^6 b a^4 b a^6 b a^4 b a^6 b
        = (a^4 b a^6 b)^3.
    \end{equation*}
\end{example}

In Theorem~\ref{thm:inj-lb},
we needed increasingly large alphabets
to give examples of words that can be mapped to increasingly high powers.
This leads to the question of whether we can do the same using an alphabet of fixed size.
More specifically,
does there exist an alphabet $\Sigma$ and a sequence of primitive words $w_1, w_2, w_3, \dotsc$
such that $|w_i|_a \ne 1$ for all $i$ and all $a \in \Sigma$
and $\lim_{i \to \infty} \max(\pex_\inj(w_i)) = \infty$?
Or can we for every alphabet $\Sigma$ find a bound $C$
such that for all primitive $w \in \Sigma^+$
either $\max(\pex_\inj(w)) \leq C$ or $\pex_\inj(w)$ is infinite?
In the case of binary alphabets, we can give such a bound $C$, and actually $C = 1$,
meaning that binary primitive words can only be mapped to primitive words,
except for the trivial cases covered by Theorem~\ref{thm:inj-inf}.
This follows from Theorem~\ref{thm:morph-prim}.
The case of ternary and larger alphabets remains open.
Results about these larger alphabets
could possibly be viewed as partial generalizations of Theorem~\ref{thm:morph-prim}.

\begin{theorem}
    Let $w \in \{a, b\}^+$ be a primitive word.
    If $|w|_a, |w|_b > 1$, then
    \begin{equation*}
        \pex_\inj(w) = \{1\}.
    \end{equation*}
\end{theorem}

\begin{proof}
    Follows from Theorem~\ref{thm:morph-prim}.
\end{proof}

As was mentioned at the beginning of the section,
the same questions we have studied for injective morphisms
could be studied also for nonperiodic morphisms.
The next example shows that
the set of words that can be mapped to arbitrarily high powers by nonperiodic morphisms
is larger and probably more complicated than
the set of words that can be mapped to arbitrarily high powers by injective morphisms.

\begin{example}
    Let $\Sigma = \{a, b, c\}$ and consider the word $w = abbacc$.
    For all $n$, we can define a morphism $h: \Sigma^* \to \Sigma^*$ by
    $h(a) = (abb)^{n - 1} a$, $h(b) = h(c) = b$.
    Then $h(w) = (abb)^{2n}$.
    Thus $2n \in \pex_\nonper(w)$.
\end{example}

\section{Algorithmic questions} \label{sec:algo}

If $A$ and $B$ are decision problems and $A$ is polynomial-time reducible to $B$,
we use the notation $A \leq_p B$.
If $A$ and $B$ are polynomially equivalent, that is, $A \leq_p B$ and $B \leq_p A$,
then we use the notation $A \equiv_p B$.

We define the following decision problems:
\begin{itemize}
    \item
    $\sat$:
    Given an alphabet $\Sigma$ and a finite system of word equations with constants over $\Sigma$,
    decide whether the system has a solution.
    \item
    $\satcf$:
    Given an alphabet $\Sigma$ and a finite system of constant-free word equations over $\Sigma$,
    decide whether the system has a nonperiodic solution.
    \item
    $\pow(n)$:
    Given an alphabet $\Sigma$ and a word $w \in \Sigma^+$,
    decide whether $n \in \gex_\nonper(w)$.
    In other words,
    decide whether $w$ can be mapped to an $n$th power by a nonperiodic morphism.
    \item
    $\nonprim$:
    Given an alphabet $\Sigma$ and a word $w \in \Sigma^+$,
    decide whether there exists $n \geq 2$ such that $ n \in \gex_\nonper(w)$.
    In other words,
    decide whether $w$ can be mapped to a nonprimitive word by a nonperiodic morphism.
\end{itemize}

In this section, we show that for all $n \geq 2$, $\pow(n) \equiv_p \satcf \equiv_p \nonprim$,
but first we need to prove a lemma and two theorems about word equations.

\begin{lemma} \label{lem:xy}
    Let $\Sigma = \{x_1, \dots, x_n\}$ and
    \begin{equation*}
        X = \prod_{i = 1}^n \prod_{j = 1}^n x_i x_j,
        \qquad
        Y = \prod_{i = 1}^n \prod_{j = 1}^n x_j x_i.
    \end{equation*}
    Let $h: \Sigma^* \to \Gamma^*$ be a morphism.
    Let $n \geq 2$ be an integer.
    The following are equivalent:
    \begin{itemize}
        \item
        $h$ is nonperiodic.
        \item
        $h(X) \ne h(Y)$.
        \item
        $h(XY) \ne h(YX)$.
        \item
        $h(X^n Y^n)$ is primitive.
    \end{itemize}
\end{lemma}

\begin{proof}
    If $h$ is periodic,
    then $h(x_i x_j) = h(x_j x_i)$ for all $i, j$,
    and thus $h(X) = h(Y)$ and $h(XY) = h(YX)$.
    Then $h(X^n Y^n) = h(X)^{2n}$ is not primitive.

    If $h$ is nonperiodic,
    then $h(x_i x_j) \ne h(x_j x_i)$ for some $i, j$,
    and thus $h(X) \ne h(Y)$ and $h(XY) \ne h(YX)$.
    Then $h(X^n Y^n)$ is primitive by the theorem of Lyndon and Sch\"utzenberger.
\end{proof}

For every finite system of constant-free word equations,
we can find a constant-free equation that has exactly the same nonperiodic solutions as the system,
as proved by Hmelevskii~\cite{hm71}.
We extend this result in two ways:
The resulting equation can be assumed to be \emph{balanced},
that is, of the form $(u, v)$ with $|u|_x = |v|_x$ for all variables $x$,
and it can be constructed in polynomial time.
This could be easily proved based on the original construction in~\cite{hm71},
but because the reference~\cite{hm71} might be hard to find and hard to read,
we give a self-contained proof with a slightly different construction.

\begin{theorem} \label{thm:balanced}
    For every finite system of constant-free equations,
    we can construct in polynomial time a balanced constant-free equation
    that has exactly the same set of nonperiodic solutions.
\end{theorem}

\begin{proof}
    Let the equations in the system be $(u_i, v_i)$ for $i \in \{1, \dots, k\}$
    and let $N$ be the length of the system.
    Let $\Sigma = \{x_1, \dots, x_n\}$ be the set of variables
    and let $X, Y$ be as in Lemma~\ref{lem:xy}.
    The equation
    \begin{equation} \label{eq:system-to-baleq}
        \Big(\prod_{i = 1}^k u_i (X^N Y^N)^2 v_i, \prod_{i = 1}^k v_i (X^N Y^N)^2 u_i\Big).
    \end{equation}
    is clearly balanced and can be constructed in polynomial time,
    and every solution of the system is a solution of this equation.
    It remains to be shown that every nonperiodic solution $h$ of~\eqref{eq:system-to-baleq}
    is a solution of the system.

    By Lemma~\ref{lem:xy}, $h(X^N Y^N)$ is primitive,
    and therefore it cannot be an internal factor of $h((X^N Y^N)^2)$.
    For all $i \in \{1, \dots, k\}$ and all $x \in \Sigma$,
    \begin{equation*}
        |u_i (X^N Y^N)^2 v_i|_x = |v_i (X^N Y^N)^2 u_i|_x
    \end{equation*}
    and therefore
    \begin{equation*}
        |h(u_i (X^N Y^N)^2 v_i)| = |h(v_i (X^N Y^N)^2 u_i)|.
    \end{equation*}
    Because $h$ is a solution of~\eqref{eq:system-to-baleq},
    \begin{equation*}
        h(u_i (X^N Y^N)^2 v_i) = h(v_i (X^N Y^N)^2 u_i)
    \end{equation*}
    for all $i$.
    Let us consider some fixed value of $i$.
    Every variable occurs in $X^N$ more times than in $v_i$, so
    \begin{equation*}
        |h(v_i X^N Y^N)| < |h((X^N Y^N)^2)| < |h(u_i (X^N Y^N)^2)|,
    \end{equation*}
    and thus $h(v_i X^N Y^N)$ is a proper prefix of $h(u_i (X^N Y^N)^2)$.
    Because $h(X^N Y^N)$ cannot be an internal factor of $h((X^N Y^N)^2)$,
    it must be $|h(v_i)| \leq |h(u_i)|$.
    Similarly, we can show that $|h(u_i)| \leq |h(v_i)|$.
    It follows that $h(u_i) = h(v_i)$.
    This holds for all $i$, so $h$ is a solution of the system.
\end{proof}

\begin{theorem} \label{thm:satcf}
    $\satcf$ is in PSPACE and NP-hard.
\end{theorem}

\begin{proof}
    The NP-hardness was proved in \cite{sa20icalp}.
    $\satcf$ being in PSPACE follows from $\satcf \leq_p \sat$, which can be seen as follows:
    Let $\Sigma = \{x_1, \dots, x_n\}$ and let $X, Y$ be as in Lemma~\ref{lem:xy}.
    Any system $S$ of constant-free equations
    can be extended in polynomial time to a system $S'$ of equations with constants
    by adding the equations $(X, xay)$ and $(Y, xbz)$,
    where $x, y, z$ are new variables and $a, b$ are distinct constant letters.
    We have to show that $S$ has a nonperiodic solution if and only if $S'$ has a solution.
    If $h$ is a nonperiodic solution of $S$,
    then $h(X) \ne h(Y)$ by Lemma~\ref{lem:xy},
    and thus and $h(X) = ua'v$ and $h(Y) = ub'w$
    for some letters $a', b'$ and words $u, v, w$.
    Every morphism we get from $h$ by renaming the letters is also a nonperiodic solution of $S$,
    so we can assume without loss of generality that $a' = a$ and $b' = b$.
    Then we can extend $h$ to a solution $\ext{h}$ of $S'$
    by defining $\ext{h}(x) = u$, $\ext{h}(y) = v$, $\ext{h}(z) = w$.
    On the other hand,
    if $g$ is a solution of $S'$, then a restriction $\res{g}$ of $g$ is a solution of $S$.
    Because
    \begin{equation*}
        \res{g}(X) = g(X) = g(x) a g(y) \ne g(x) b g(z) = g(Y) = \res{g}(Y),
    \end{equation*}
    it follows from Lemma~\ref{lem:xy} that $\res{g}$ is nonperiodic.
    This completes the proof.
\end{proof}

Now we can start analyzing the problems $\pow(n)$ and $\nonprim$.

\begin{theorem} \label{thm:pow}
    Let $n \geq 2$ be an integer.
    Then $\pow(n) \equiv_p \satcf$.
\end{theorem}

\begin{proof}
    First, we show that $\pow(n) \leq_p \satcf$.
    Given an alphabet $\Sigma$ and a word $w \in \Sigma^+$,
    we can construct in polynomial time
    a constant-free word equation $(w, x^n)$ over $\Sigma \cup \{x\}$,
    where $x \notin \Sigma$ is a new letter.
    By Theorem~\ref{thm:gex-eq},
    $n \in \gex_\nonper(w)$ if and only if this equation has a nonperiodic solution.

    Then, we show that $\satcf \leq_p \pow(n)$.
    Given a system of constant-free word equations,
    we can use Theorem~\ref{thm:balanced} to construct a balanced constant-free equation $(u, v)$
    that has exactly the same set of nonperiodic solutions.
    We can show that $n \in \gex_\nonper(uv^{n - 1})$
    if and only if $(u, v)$ has a nonperiodic solution as follows.
    If $n \in \gex_\nonper(uv^{n - 1})$,
    then $h(uv^{n - 1}) = s^n$ for some nonperiodic morphism and word $s$.
    From $(u, v)$ being balanced it follows that $|h(u)| = |h(v)|$,
    and then $|h(u)| = |h(v)| = |s|$ and thus $h(u) = s = h(v)$,
    meaning that $h$ is a nonperiodic solution of the equation $(u, v)$.
    On the other hand, if $h$ is a nonperiodic solution of $(u, v)$,
    then $h(uv^{n - 1}) = h(u)^n$,
    meaning that $n \in \gex_\nonper(uv^{n - 1})$
    or $h(u) = \eps$.
    But if $h(u) = \eps$,
    then it must be $h(a) = \eps$ for all $a \in \alphabet(u) = \alphabet(v)$,
    which means that $h$ being nonperiodic
    implies that $\Sigma \smallsetminus \alphabet(u) \ne \varnothing$,
    and then $n \in \gex_\nonper(uv^{n - 1})$ by Theorem~\ref{thm:nonper-sigma}.
\end{proof}

\begin{theorem} \label{thm:nonprim}
    We have $\nonprim \equiv_p \satcf$.
\end{theorem}

\begin{proof}
    First, we show that $\nonprim \leq_p \satcf$.
    Given an alphabet $\Sigma$ and a word $w \in \Sigma^+$,
    we can construct in polynomial time
    a system of two constant-free word equations $\{(w, x^2 y^3), (xy, yx)\}$
    over $\Sigma \cup \{x, y\}$,
    where $x, y \notin \Sigma$ are new letters.
    By Theorem~\ref{thm:nonprim-eq},
    there exists $n \geq 2$ such that $n \in \gex_\nonper(w)$
    if and only if this system has a nonperiodic solution.

    Then, we show that $\satcf \leq_p \nonprim$.
    Given a system of constant-free word equations,
    we can use Theorem~\ref{thm:balanced} to construct a balanced constant-free equation $(u, v)$
    that has exactly the same set of nonperiodic solutions.
    Let $\Sigma = \{x_1, \dots, x_n\}$ be the set of variables
    and let $X, Y$ be as in Lemma~\ref{lem:xy}.
    Let
    \begin{equation*}
        Z = X^{|uv|} Y^{|uv|}.
    \end{equation*}
    We can show that there exists $n \in \gex_\nonper(Z^4 u Z^4 v)$, $n \geq 2$,
    if and only if $(u, v)$ has a nonperiodic solution as follows.
    If $h$ is a nonperiodic solution of $(u, v)$, then $h(Z^4 u Z^4 v) = h(Z^4 u)^2$,
    and $h(Z^4 u) \ne \eps$ because every letter occurs in $Z$,
    so $2 \in \gex_\nonper(Z^4 u Z^4 v)$.
    On the other hand, if $n \in \gex_\nonper(Z^4 u Z^4 v)$, $n \geq 2$, then
    \begin{equation*}
        h(Z^4 u Z^4 v) = r^n
    \end{equation*}
    for some nonempty word $r$ and $h \in \nonper$.
    If $n = 2$, then $h(Z^4 u) = r = h(Z^4 v)$ and thus $h(u) = h(v)$.
    If $n \geq 3$, then $h(Z^4)$ and $r^n$ have a common prefix of length
    \begin{align*}
        |h(Z^4)|
        &= |h(Z)| + |h(Z^9)| / 3\\
        &\geq |h(Z)| + |h(Z^8) h(uv)| / 3
        = |h(Z)| + |r^n| / 3\\
        &\geq |h(Z)| + |r|.
    \end{align*}
    From the theorem of Fine and Wilf it follows that
    $h(Z)$ and $r$ have the same primitive root.
    By Lemma~\ref{lem:xy}, $h(Z)$ is primitive, so $r \in h(Z)^+$.
    It follows that
    \begin{equation*}
        h(u) h(Z)^4 h(v) = h(Z)^m
    \end{equation*}
    for some integer $m$.
    Because $h(Z)$ cannot be an internal factor of $h(Z)^2$, it must be $h(u), h(v) \in h(Z)^*$.
    Because $(u, v)$ is balanced, $|h(u)| = |h(v)|$ and thus $h(u) = h(v)$.
\end{proof}

\begin{corollary}
    Let $n \geq 2$ be an integer.
    Then $\pow(n)$ and $\nonprim$ are in PSPACE and NP-hard.
\end{corollary}

\begin{proof}
    Follows from Theorems~\ref{thm:satcf}, \ref{thm:pow}, \ref{thm:nonprim}.
\end{proof}

We could also study the problems of deciding
whether $n \in \gex_\inj(w)$ or whether $n \in \pex_\nonper(w)$.
The property of being primitive cannot be expressed by word equations~\cite{kamipl00},
which means that $\pex_\nonper(w)$ might be more complicated than $\gex_\nonper(w)$
from an algorithmic point of view.

\section{Conclusion}

In this article, we have characterized the words $w$ such that $\pex_\inj(w)$ is infinite,
and we have proved upper and lower bounds for how big $\max(\pex_\inj(w))$ can be
for other words $w$.
Proving bounds for fixed alphabets remains an open problem.
Similar questions could be studied for $\nonper$ in place of $\inj$.

We have also proved that the algorithmic problems of determining
whether $n \in \gex_\nonper(w)$
or whether $\gex_\nonper(w) \ne \{1\}$ for a given word $w$
are closely connected to the satisfiability problem of word equations,
leading to upper and lower bounds for the complexity of these problems.
Similar questions could be studied for $\pex$ in place of $\gex$
or for $\inj$ in place of $\nonper$.

\bibliography{ref}

\end{document}